\documentclass[a4paper]{amsart} 
\usepackage[utf8]{inputenc}
\usepackage[english]{babel}
\usepackage[T1]{fontenc}

\usepackage{latexsym,amsmath,amsfonts,amssymb,amsthm,amstext,textcomp,mathabx,relsize}
\usepackage{mathrsfs,dutchcal,bbm}
\usepackage{enumerate}
\usepackage{hyperref}
\usepackage{filecontents}

\usepackage[capitalize]{cleveref}

\makeatletter
\def\@seccntformat#1{%
	\protect\textup{\protect\@secnumfont
		\ifnum\pdfstrcmp{subsection}{#1}=0 \bfseries\fi
		\ifnum\pdfstrcmp{subsubsection}{#1}=0 \itshape\fi
		\csname the#1\endcsname
		\protect\@secnumpunct
	}%
}
\renewcommand{\@upn}{}
\DeclareRobustCommand{\crefnosort}[1]{%
	\begingroup\@cref@sortfalse\cref{#1}\endgroup
}
\makeatother

\newcommand{\EE}{\mathbb{E}}
\newcommand{\CC}{\mathbb{C}}

\newcommand{\PP}{\mathbb{P}}

\newcommand{\RR}{\mathbb{R}}

\newcommand{\ii}{\mathrm{i}}
\newcommand{\eul}{\mathrm{e}}
\newcommand{\sym}{\mathrm{sym}}

\newcommand{\Id}{\mathrm{d}}

\newcommand{\LO}{\mathcal{B}}
\newcommand{\dom}{\mathcal{D}}

\newcommand{\Fock}{\mathcal{F}}
\newcommand{\UV}{\Lambda}
\newcommand{\ad}{a^\dagger}
\newcommand{\id}{\mathbbm{1}}

\newcommand{\expv}[1]{\epsilon(#1)}
\newcommand{\p}{\mathrm{p}}
\newcommand{\bos}{\mathrm{b}}
\newcommand{\ren}{\mathrm{ren}}
\newcommand{\vp}{\varphi}

\newcommand{\wt}[1]{\widetilde{#1}}

\newcommand{\wh}[1]{\widehat{#1}}  
 
\newcommand{\fr}[1]{\mathfrak{#1}}
\newcommand{\mc}[1]{\mathcal{#1}}
\newcommand{\scr}[1]{\mathscr{#1}}
\renewcommand{\Im}{\mathrm{Im}}
\renewcommand{\Re}{\mathrm{Re}}
\renewcommand{\le}{\leqslant}
\renewcommand{\ge}{\geqslant}
\theoremstyle{plain}
\newtheorem{thm}{Theorem}[section]

\newtheorem{lem}[thm]{Lemma}
\newtheorem{cor}[thm]{Corollary}
\newtheorem{prop}[thm]{Proposition}
\theoremstyle{definition}

\theoremstyle{remark}
 
\newtheorem{rem}[thm]{Remark}
\numberwithin{equation}{section}
\crefname{equation}{}{}
\Crefname{equation}{}{}
\crefname{enumi}{}{}
\Crefname{enumi}{}{}
\crefname{lem}{Lemma}{Lemmas}
\Crefname{lem}{Lemma}{Lemmas}
\crefname{thm}{Theorem}{Theorems}
\Crefname{thm}{Theorem}{Theorems}
\crefname{prop}{Proposition}{Propositions}
\Crefname{prop}{Proposition}{Propositions}
\crefname{defn}{Definition}{Definitions}
\Crefname{defn}{Definition}{Definitions}

\title[Feynman--Kac formula for the relativistic Nelson model]{Feynman--Kac formula for fiber hamiltonians in the relativistic Nelson model in two spatial dimensions}
\author{Benjamin Hinrichs}
\address{Benjamin Hinrichs, Universit\"at Paderborn, Institut f\"ur Mathematik, Institut f\"ur Photonische Quantensysteme, Warburger Str. 100, 33098 Paderborn, Germany}
\email{benjamin.hinrichs@math.upb.de}

\author{Oliver Matte}
\address{Oliver Matte, Aalborg Universitet, Institut for Matematiske Fag, Skjernvej 4a, 9220 Aalborg, Denmark}
\email{oliver@math.aau.dk}
\usepackage{xcolor}\usepackage{cancel}

\begin{document}

\begin{abstract} 
	\noindent  
	In this proceeding we consider a translation invariant Nelson type model in two spatial
	dimensions modeling a scalar relativistic particle in interaction with a massive radiation field.
	As is well-known, the corresponding Hamiltonian can be defined with the help of an energy renormalization.
	First, we review a Feynman--Kac formula for the semigroup generated by this Hamiltonian proven by the
	authors in a recent preprint (where several matter particles and exterior potentials are treated as well).
	After that, we employ a few technical key relations and estimates obtained in our preprint to present an otherwise
	self-contained derivation of new Feynman--Kac formulas
	for the fiber Hamiltonians attached to fixed total momenta of the translation invariant system. We conclude by
	inferring an alternative derivation of the Feynman--Kac formula for the full translation invariant Hamiltonian.
\end{abstract}

\maketitle

\section{Introduction}

\noindent
The original Nelson model describes a conserved number of non-relativistic quantum mechanical matter
particles linearly coupled to a quantized radiation field (boson field). 
Its crucial feature is its comparatively simple renormalizability.
In fact, the heuristic matter-radiation interaction term in the Hamiltonian is {\it a priori} ill-defined, as its behavior for
large boson momenta is too singular. Imposing an ultraviolet cutoff in the interaction term 
and adding explicitly given cutoff dependent
renormalization energies, we obtain, however, a well-defined family of Hamiltonians converging in the norm resolvent
sense to a unique {\em renormalized} Hamiltonian as the cutoff parameter goes to infinity. This has been demonstrated by
Nelson almost sixty years ago \cite{Nelson.1964,Nelson.1964c} and in later technical improvements by several authors
(norm instead of strong resolvent convergence has been observed first by Ammari \cite{Ammari.2000}).
Ever since the spectral and probabilistic analysis of Nelson's model and variants thereof has been a popular topic
in mathematical quantum field theory.

A modification of Nelson's model, where the (scalar) matter particles are relativistic as well, has already been
studied in the beginning of the 1970's. Working in spatial dimension three, as Nelson did, Gross
\cite{Gross.1973} was able to prove the existence of renormalized Hamiltonians by
procedures more elaborate than Nelson's involving implicit particle mass renormalizations, a passage
to a non-Fock representation and compactness arguments; whether Gross' renormalized Hamiltonian
is unique is still unclear. Sloan \cite{Sloan.1974} treated the relativistic version of Nelson's model in spatial dimension two
and was able to prove resolvent convergence, in the strong sense and along a subsequence of 
a given sequence of cutoff parameters at least. A few years ago only, Schmidt \cite{Schmidt.2019} 
provided a new construction of Sloan's renormalized Hamiltonian. Employing the recently developed 
method of interior boundary conditions (see \cite{LampartSchmidt.2019} and the references therein), 
Schmidt proved proper norm resolvent convergence and 
obtained explicit formulas for the domain of the renormalized Hamiltonian and its action on it.

In our work we are interested in the probabilistic analysis of Nelson type models and in particular in deriving
Feynman--Kac formulas for the semigroups generated by the (semibounded) renormalized Hamiltonians.
While Nelson obtained probabilistic representations of certain matrix elements of the unitary group
\cite{Nelson.1964c}, Feynman--Kac formulas for the semigroup in the original Nelson model were
found in \cite{GubinelliHiroshimaLorinczi.2014,MatteMoller.2018}. The methods used in \cite{MatteMoller.2018}
also apply {\em mutatis mutandis} to the Hamiltonian constructed by Sloan and Schmidt, that we refer to as the
relativistic Nelson Hamiltonian in two spatial dimensions. For the $N$ matter particle version of that
Hamiltonian including exterior potentials, the present authors proved Feynman--Kac formulas in \cite{HinrichsMatte.2023}.
For earlier work on Feynman--Kac formulas for ultraviolet regularized Nelson type and related models
(such as the Pauli-Fierz model) and numerous applications we refer to the textbook 
\cite{HiroshimaLorinczi.2020} and the references given there.

In this proceeding we discuss the translation invariant (no external potential) relativistic Nelson Hamiltonian 
for one matter particle in two space dimensions. This Hamiltonian is unitarily equivalent to a direct integral
of fiber Hamiltonians, each attached to a fixed total momentum of the entire matter-radiation system.
After reviewing the Feynman--Kac formula for the full Hamiltonian from \cite{HinrichsMatte.2023}, we shall derive
Feynman--Kac formulas for the fiber Hamiltonians, employing only a few key estimates and relations from
\cite{HinrichsMatte.2023} as starting points. It would also be possible to explicitly fiber-decompose the
probabilistic side of the Feynman--Kac formula for the full Hamiltonian and spend a little bit of work to
argue that the so-obtained families of operators define a semigroup for every fixed total momentum, 
that must be generated by a corresponding renormalized fiber Hamiltonian; compare, e.g.,
\cite[Chapter~7]{MatteMoller.2018} for the non-relativistic case. Here we favor, however,
the more independent derivation only based on the technical key inputs from \cite{HinrichsMatte.2023}.
For in this way, the reader can see proof strategies from \cite{HinrichsMatte.2023} at work in a 
slightly different setting. 

\subsection*{Structure of the proceeding}

\noindent
After using the remaining part of this introduction to
clarify our notation for operators in bosonic Fock space,
we briefly explain the construction of the Hamiltonian $H$ for the translation invariant relativistic Nelson 
model in two spatial dimensions in \cref{sec:Ham}. In \cref{sec:FKH} we introduce some stochastic processes
employed throughout the proceeding and present Feynman--Kac formulas found in \cite{HinrichsMatte.2023}
for the semigroups generated by $H$ and its versions $H_{\UV}$ containing ultraviolet cutoff interaction terms. 
By means of a Lee-Low-Pines transformation, we shall turn $H_{\UV}$ into
a direct integral of fiber Hamiltonians $\wh{H}_{\UV}(\xi)$ attached to total momenta $\xi\in\RR^2$ of the 
matter-radiation system in \cref{sec:LLP}.
The objective of \cref{sec:keys} is to present the crucial technical ingredients
from \cite{HinrichsMatte.2023} applied in the remaining parts of the text, which otherwise are fairly self-contained.
Our derivation of Feynman--Kac formulas for fiber Hamiltonians starts in \cref{sec:FKint},
where the corresponding Feynman--Kac integrands and semigroups are analyzed first.
The Feynman--Kac formulas themselves are established in \cref{sec:FKxi}, first for the ultraviolet regularized
fiber operators $\wh{H}_{\UV}(\xi)$ and afterwards for their renormalized versions, i.e.,
the norm resolvent limits $\wh{H}(\xi)\coloneq\lim_{\UV\to\infty }\wh{H}_{\UV}(\xi)$. 
In fact, as a byproduct of our method, we shall obtain an independent proof for the existence of these limits,
improving on \cite{Sloan.1974} where only strong resolvent convergence along subsequences is proven.
(While Schmidt treated the full Hamiltonian explicitly in \cite{Schmidt.2019}, norm resolvent convergence of fiber Hamiltonians
can probably be inferred from his results, too, see \cite{DamHinrichs.2021} for an approach along these lines.)
Clearly as expected, it turns out that $H$ is the direct integral of the renormalized fiber operators $\wh{H}(\xi)$
after a Lee-Low-Pines transformation. This is verified in \cref{sec:renrev}, where we also present an
alternative derivation of the Feynman--Kac formulas for the full translation invariant Hamiltonians, based on the ones for
fiber operators.

\subsection*{Operators in bosonic Fock space}

\noindent
All fiber Hamiltonians alluded to above act in the bosonic Fock space over $L^2(\RR^2)$ defined by
\begin{align*}
\Fock&\coloneq \CC\oplus\bigoplus_{n=1}^\infty L^2_{\sym}(\RR^{2n}).
\end{align*}
In the above relation $L^2_{\sym}(\RR^{2n})$ is the closed subspace of $L^2(\RR^{2n})$ comprising
all functions $\phi_n\in L^2(\RR^{2n})$ satisfying $\phi_n(k_1,\ldots,k_n)=\phi_n(k_{\pi(1)},\ldots,k_{\pi(n)})$
a.e. for every permutation $\pi$ of $\{1,\ldots,n\}$; here $k_j\in\RR^2$ for every $j\in\{1,\ldots,n\}$. 

Throughout this proceeding we use standard notation for the following operators acting in $\Fock$
(see, e.g., \cite{Arai.2018,Parthasarathy.1992} for their constructions and basic properties):

For every $f\in L^2(\RR^2)$, the symbols $\ad(f)$ and $\vp(f)$ denote the
corresponding creation and field operators, respectively. 
Thus, $\vp(f)$ is selfadjoint and equal to the closure of
$\ad(f)+\ad(f)^*$. If $V$ is a unitary operator on $L^2(\RR^2)$, then $\Gamma(V)$ denotes its second
quantization, which is a unitary operator on $\Fock$. We shall sometimes use that
$\Gamma(V_1)\Gamma(V_2)=\Gamma(V_1V_2)$ for  unitary operators $V_1$ and $V_2$ on $L^2(\RR^2)$.
Finally, if $A$ is a selfadjoint operator on $L^2(\RR^2)$, then $\Id\Gamma(A)$ denotes its differential
second quantization. That is, $\Id\Gamma(A)$ is the unique selfadjoint operator on $\Fock$
satisfying $\eul^{-\ii t\Id\Gamma(A)}=\Gamma(\eul^{-\ii tA})$, $t\in\RR$.

\section{The relativistic Nelson model in two spatial dimensions}\label{sec:Ham}

\noindent
In this section we first introduce the Hamiltonian $H_{\UV}$ for the total particle-radiation system
with an ultraviolet cutoff interaction term and finally the renormalized Hamiltonian $H$. Both are
selfadjoint operators in the Hilbert space $L^2(\RR^2,\Fock)$.

The matter particle is assumed to have a non-negative mass $m_{\p}\ge0$ and dispersion relation
\begin{align}\label{defpsi}
\psi(\xi)&\coloneq(|\xi|^2+m_{\p}^2)^{1/2}-m_{\p},\quad\xi\in\RR^2.
\end{align}
Since the model would be unstable otherwise, the bosons have a strictly positive mass $m_{\bos}>0$.
The dispersion relation for a single boson is
\begin{align}\label{defomega}
\omega(k)&\coloneq (|k|^2+m_{\bos}^2)^{1/2},\quad k\in\RR^2.
\end{align}
The coupling function for the matter-radiation interaction is given by
\begin{align}\label{defv}
v&\coloneq g\omega^{-1/2}\notin L^2(\RR^2),\quad\text{with a coupling constant $g\in\RR\setminus\{0\}$.}
\end{align}

Since $v$ is not square-integrable, an energy renormalization will be necessary to define the
Hamiltonian $H$ for our model. That is, we first introduce Hamiltonians containing the ultraviolet
cutoff coupling functions
\begin{align*}
v_{\UV}\coloneq \chi_{B_{\UV}}v\in L^2(\RR^2),\quad \UV\in[0,\infty).
\end{align*}
Here $\chi_{B_{\UV}}$ is the indicator function of the two-dimensional open ball of radius $\UV$ about 
the origin $B_{\UV}$, with the understanding that $B_0=\emptyset$.
Abbreviating
\begin{align*}
e_x(k)&\coloneq \eul^{-\ii k\cdot x},\quad k\in\RR^2,
\end{align*}
for every $x\in\RR^2$, and introducing renormalization energies
\begin{align}\label{defEren}
E_{\UV}^{\ren}&\coloneq \int_{B_{\UV}}\frac{v^2(k)}{\omega(k)+\psi(k)}\Id k,\quad \UV\in[0,\infty),
\end{align}
we define the relativistic Nelson operator with ultraviolet
cutoff at $\UV\in[0,\infty)$ by
\begin{align*}
(H_{\UV}\Phi)(x)&\coloneq (\psi(-\ii\nabla)\Phi)(x)+\Id\Gamma(\omega)\Phi(x)
+\vp(e_xv_{\UV})\Phi(x)+E_{\UV}^{\ren}\Phi(x),
\end{align*}
for a.e. $x\in\RR^2$ and every 
\begin{align*}
\Phi\in\dom(H_\UV)=\dom(H_0)\coloneq H^1(\RR^2,\Fock)\cap L^2(\RR^2,\dom(\Id\Gamma(\omega))).
\end{align*}
Here and henceforth, $\dom(\cdot)$ stands for domains of selfadjoint operators equipped with their graph norms.
The Sobolov space $H^1(\RR^2,\Fock)$ is defined via the $\Fock$-valued Fourier transformation $F$,
that is given by Bochner-Lebesgue integrals
\begin{align*}
(F\Phi)(\xi)&\coloneq\frac{1}{2\pi}\int_{\RR^2}\eul^{-\ii\xi\cdot x}\Phi(x)\Id x,\quad \xi\in\RR^2,
\end{align*}
whenever $\Phi\in L^1(\RR^2,\Fock)\cap L^2(\RR^2,\Fock)$, and isometric extension to $L^2(\RR^2,\Fock)$.
The selfadjoint operator $\psi(-\ii\nabla)$ is defined by means of $F$ as well, i.e., by definition,
\begin{align*}
(F\psi(-\ii\nabla)\Phi)(\xi)&=\psi(\xi)(F\Phi)(\xi),\quad\text{a.e. $\xi\in\RR^2$,}
\end{align*}
for every $\Phi\in\dom(\psi(-\ii\nabla))=H^1(\RR^2,\Fock)$. Employing the 
Kato-Rellich theorem and the standard relative bound
\begin{align}\label{rbvp}
\|\vp(e_xv_{\UV})\phi\|&\le2^{1/2} \|(\omega^{-1/2}\vee1)v_{\UV}\|\|(1+\Id\Gamma(\omega))^{1/2}\phi\|,
\quad x\in\RR^2,
\end{align}
which is available for all $\phi$ in the form domain of $\Id\Gamma(\omega)$, we can indeed verify selfadjointness of 
every $H_{\UV}$ with $\UV\in(0,\infty)$ on $\dom(H_0)$.

Finally, the renormalized relativistic Nelson operator in two space dimensions is given by
\begin{align}\label{reslimHUV}
H&\coloneq H_{\infty}\coloneq\underset{\UV\to\infty}{\textrm{norm-resolvent-lim}} \ H_\UV.
\end{align}
Existence of the above limit has been established in \cite{Sloan.1974,Schmidt.2019}. It has been
re-proven in \cite{HinrichsMatte.2023} as an automatic byproduct of the proof strategy for the Feynman--Kac 
formula established there; see \cref{sec:renrev} for yet another proof.

\section{Feynman--Kac formulas for the full Hamiltonians}\label{sec:FKH}

\noindent
Throughout this proceeding we fix a filtered probability space $(\Omega,\fr{F},(\fr{F}_t)_{t\ge0},\PP)$
satisfying the usual hypotheses as well as a $(\fr{F}_t)_{t\ge0}$-L\'{e}vy process $X$
whose L\'{e}vy symbol is $-\psi$ and all whose paths are c\`{a}dl\`{a}g.
Expectations with respect to $\PP$ will be denoted by $\EE$, and we
put $X_{t-}\coloneq \lim_{s\uparrow t}X_s$ for all $t>0$. We recall that $X$ has characteristics
$(0,0,\nu)$, where its L\'{e}vy measure $\nu$ has an explicitly known density with respect to the Lebesgue-Borel measure;
see, e.g., \cite[\textsection2.2]{HinrichsMatte.2023}.

Next, we define the stochastic processes appearing in our Feynman--Kac integrands:
For every $\UV\in[0,\infty]$, we introduce the following well-defined $L^2(\RR^2)$-valued
Bochner-Lebesgue integrals,
\begin{align}\label{defUplusminus}
U_{\UV,t}^-&\coloneq\int_0^t\eul^{-s\omega}e_{X_s}v_{\UV}\Id s,\quad
U_{\UV,t}^+\coloneq\int_0^t\eul^{-(t-s)\omega}e_{X_s}v_{\UV}\Id s,\quad t\ge0.
\end{align}
For both choices of the sign, $(U_{\UV,t}^\pm)_{t\ge0}$ is a continuous and adapted
$L^2(\RR^2)$-valued process \cite[Appendix~B]{HinrichsMatte.2023}.
For finite $\UV$, the analogue of Feynman's complex action in our model is given by
\begin{align}\label{defuUV}
u_{\UV,t}&\coloneq \int_0^t\langle U_{\UV,s}^+|e_{X_s}v_{\UV}\rangle\Id s-tE_{\UV}^{\ren},
\quad t\ge0,\,\UV\in[0,\infty).
\end{align}
It defines a real-valued continuous and adapted process. In \cite{HinrichsMatte.2023} and, in a slightly more sketchy fashion,
in \cref{lem:rewriteuUV}, we re-write this
expression employing It\^{o}'s formula and obtain a more regular one where the ultraviolet cutoff can be dropped.
This results in the following formula for the limiting complex action: Setting
\begin{align*}
\beta&\coloneq(\omega+\psi)^{-1}v\in L^2(\RR^2),
\end{align*}
we define
\begin{align}\label{def:action}
u_{\infty,t}&\coloneq \int_{(0,t]\times\RR^2}\langle U_{\infty,s}^{+}|e_{X_{s-}}(e_z-1)\beta\rangle\Id\wt{N}(s,z)
-\langle U_{\infty,t}^{+}|e_{X_t}\beta\rangle,\quad t\ge0.
\end{align}
Here the integral is an isometric stochastic integral with respect to the martingale valued measure
$\wt{N}$ of $X$; see, e.g., \cite{Applebaum.2009} for the nomenclature used here and detailed explanations.
The corresponding stochastic integral process is a c\`{a}dl\`{a}g $L^2$-martingale. 
(In fact, the paths of $u_\infty$ are $\PP$-a.s. continuous \cite[Corollary~6.9]{HinrichsMatte.2023}.)

The last building blocks in our Feynman--Kac integrands are the operator norm convergent series
\begin{align*}
F_t(h)&\coloneq \sum_{n=0}^\infty \frac{1}{n!}\ad(h)^n\eul^{-t\Id\Gamma(\omega)},\quad t>0,\,h\in L^2(\RR^2),
\end{align*}
which define analytic maps $F_t:L^2(\RR^2)\to \LO(\Fock)$. For these maps and their derivatives we have the bounds
\cite{GueneysuMatteMoller.2017}
\begin{align}\label{bdFt}
\|F_t(h)\|&\le \mc{S}(\|h\|_t),\quad \|F_t'(h)\tilde{h}\|\le 4\|\tilde{h}\|_t\mc{S}(\|h\|_t),\quad h,\tilde{h}\in L^2(\RR^2),
\end{align}
where $\|h\|_t^2\coloneq\|h\|^2+\|(t\omega)^{-1/2}h\|^2$ and $\mc{S}(z)\coloneq\sum_{n=0}^\infty(n!)^{-1/2}(2z)^n$, $z\in\CC$.

We are now in a position to introduce the Fock space operator-valued parts of the Feynman--Kac integrands
for the entire matter-radiation system. For all $\UV\in[0,\infty]$ and $x\in\RR^2$, they are given by the adjoints of
\begin{align}\label{def:Wt}
W_{\UV,t}(x)&\coloneq \eul^{u_{\UV,t}}F_{t/2}(-e_xU_{\UV,t}^{+})F_{t/2}(-e_xU_{\UV,t}^{-})^*
=\Gamma(e_{x})W_{\UV,t}(0)\Gamma(e_{-x}),
\end{align}
whenever $t>0$, and $W_{\UV,0}(x)\coloneq \id_{\Fock}$.

The next theorem is a special case of \cite[Theorem~2.1]{HinrichsMatte.2023}.
Departing from a few technical key ingredients presented in \cref{sec:keys}
we shall obtain an otherwise independent proof of 
the asserted formula \cref{eq:FK} (for a.e. $x$) at the end of \cref{sec:renrev}.

\begin{thm}[{\bf Feynman--Kac formulas for the entire system}]
Let $\UV\in[0,\infty]$, $\Phi\in L^2(\RR^{2},\Fock)$ and $t>0$. Then $\eul^{-tH_{\UV}}\Phi$ 
has a unique continuous representative which is given by
\begin{align}\label{eq:FK}
(\eul^{-tH_{\UV}}\Phi)(x)&=\EE\left[W_{\UV,t}(x)^*\Phi(x+X_t)\right],\quad x\in\RR^2.
\end{align}
\end{thm}

\section{Lee-Low-Pines transformation and fiber Hamiltonians}\label{sec:LLP}

\noindent
The Hamiltonians $H_{\UV}$ and $H$ are invariant under translations of the entire matter-radiation 
system in space and can therefore be represented as direct integrals with respect to the system's total momentum
of selfadjoint fiber Hamiltonians. This is implemented by the Lee-Low-Pines transformation in two dimensions defined by
\begin{align}\label{defLLP}
\scr{U}&\coloneq F\int_{\RR^2}^\oplus\Gamma(e_{-x})\Id x.
\end{align}

We shall briefly discuss the transformation by $\scr{U}$ of the ultraviolet regularized Hamiltonians
$H_{\UV}$ with $\UV\in[0,\infty)$:

For $i\in\{1,2\}$, we let $K_i$ denote the maximal operator of multiplication with 
$k_i$ on $L^2(\RR^2)$, i.e., $(K_if)(k)=k_if(k)$, a.e. $k\in\RR^2$, $f\in \dom(K_i)$.
Further, we put $\Id\Gamma(K)\coloneq(\Id\Gamma(K_1),\Id\Gamma(K_2))$ and
$\dom(\Id\Gamma(K))\coloneq \bigcap_{i=1}^2\dom(\Id\Gamma(K_i))$. Then the fiber
Hamiltonian $\wh{H}_{\UV}(\xi)$ with $\UV\in[0,\infty)$ attached to the total
momentum $\xi\in\RR^2$ turns out to be
\begin{align*}
\wh{H}_{\UV}(\xi)&\coloneq \psi(\xi-\Id\Gamma(K))+\Id\Gamma(\omega)+\vp(v_{\UV})+E_{\UV}^{\ren}.
\end{align*}
In view of \eqref{rbvp} (with $x=0$) this operator is selfadjoint on 
$\dom(\wh{H}_{\UV}(\xi))=\dom(\Id\Gamma(\omega))$.

In fact, it is straightforward to verify that 
\begin{align*}
\scr{U}H^1(\RR^2,\Fock)=\bigg\{\Psi\in L^2(\RR^2,\Fock)\,\bigg|&\,\Psi(\xi)\in\dom(\Id\Gamma(K))\,\;\text{a.e. $\xi$, and}
\\
&\int_{\RR^2}\|\psi(\xi-\Id\Gamma(K))\Psi(\xi)\|^2\Id\xi<\infty\bigg\},
\end{align*}
and, for all $\Phi\in H^1(\RR^2,\Fock)$,
\begin{align*}
(\scr{U}\psi(-\ii\nabla)\Phi)(\xi)&=\psi(\xi-\Id\Gamma(K))(\scr{U}\Phi)(\xi),\quad \text{a.e. $\xi\in\RR^2$.}
\end{align*}
Moreover, $\scr{U}$ maps $L^2(\RR^2,\dom(\Id\Gamma(\omega)))$ into itself and 
\begin{align*}
\scr{U}\Id\Gamma(\omega)\Phi=\Id\Gamma(\omega)\scr{U}\Phi,\quad\Phi\in L^2(\RR^2,\dom(\Id\Gamma(\omega))).
\end{align*}
Finally, we have the well-known commutation relations
\begin{align}\label{comrelGammavp}
\Gamma(e_{-x})\vp(e_xv_{\UV})\phi&=\vp(v_{\UV})\Gamma(e_{-x})\phi,\quad x\in\RR^2,\,\UV\in[0,\infty),
\end{align}
for, e.g., all $\phi\in\dom(\Id\Gamma(\omega))$. 
Putting these remarks together and observing strong resolvent measurability of 
the family $(\wh{H}_{\UV}(\xi))_{\xi\in\RR^2}$, we infer indeed that
\begin{align}\label{fibdecHUV}
\scr{U}H_{\UV}\scr{U}^*=\int_{\RR^2}^\oplus\wh{H}_{\UV}(\xi)\Id\xi,\quad\UV\in[0,\infty).
\end{align}
An analogous relation for $H$ is derived in \cref{thm:fibH} below.

\section{Main technical ingredients}\label{sec:keys}

\noindent
In this section we collect the main technical ingredients from \cite{HinrichsMatte.2023} that
we shall employ in the remaining part of this proceeding to give an otherwise fairly self-contained
derivation of Feynman--Kac formulas for fiber Hamiltonians.

We start by explaining where our formula \cref{def:action} for the 
complex action $u_{\infty,t}$ originates from. Notice that both terms in the definition
\cref{defuUV} of $u_{\UV,t}$ with finite $\UV$ become ill-defined when the cutoff at $\UV$ is dropped.
We can, however, exploit the presence of the oscillating terms $e_{X_s}$ under the integral
in \cref{defuUV} to arrive at a new formula for $u_{\UV,t}$ comprising more regular terms.
This is done with the help of It\^{o}'s formula:

\begin{lem}\label{lem:rewriteuUV}
Let $\UV\in[0,\infty)$. Then, $\PP$-a.s.,
\begin{align}\label{foruUVIto}
u_{\UV,t}&=\int_{(0,t]\times\RR^2}
\langle U^+_{\UV,s}|e_{X_{s-}}(e_z-1)\beta\rangle \Id\wt{N}(s,z)
-\langle U^+_{\UV,t}|e_{X_t}\beta\rangle,\quad t\ge0.
\end{align}
\end{lem}

As the reader will notice, in \cref{def:action} we turn the identity \cref{foruUVIto} satisfied
for finite $\UV$ into a definition of $u_{\infty,t}$. Recall that $U^+_{\infty,t}$
is well-defined right away. Further, it is not difficult to check that the stochastic integral
in \cref{foruUVIto} is meaningful for $\UV=\infty$ as well.

\begin{proof}[Sketch of the proof of \cref{lem:rewriteuUV}.]
The first step is to observe the integral equation
\begin{align}\label{IEUplus}
U_{\UV,t}^+&=\int_0^t(e_{X_s}v_{\UV}-\omega U^+_{\UV,s})\Id s,\quad t\ge0,
\end{align}
which can be derived with the help of \cref{defUplusminus} and
the fundamental theorem of calculus for the Lebesgue integral \cite[Lemma~4.1]{HinrichsMatte.2023}.
Since $\RR^2\ni z\mapsto \chi_{B_{\UV}}\eul^{-\ii K\cdot z}\beta\in L^2(\RR^2)$ is bounded and smooth
with bounded partial derivatives of any order, and since
$U^+_{\UV,t}=\chi_{B_{\UV}}U^+_{\UV,t}$, we can combine \cref{IEUplus} with It\^{o}'s formula to $\PP$-a.s. get
\begin{align*}
\langle U^+_{\UV,t}|e_{X_t}\beta\rangle
&=\int_0^t\langle e_{X_s}v_{\UV}|e_{X_s}\beta\rangle\Id s
\\
&\quad-\int_0^t\langle \omega U^+_{\UV,s}|e_{X_s}\beta\rangle\Id s
-\int_0^t\langle \psi U^+_{\UV,s}|e_{X_s}\beta\rangle\Id s
\\
&\quad +\int_{(0,t]\times\RR^2}
\langle U^+_{\UV,s}|e_{X_{s-}}(e_{z}-1)\beta\rangle \Id\wt{N}(s,z),\quad t\ge0.
\end{align*}
Here the integral in the first line of the right hand side equals $tE_{\UV}^{\ren}$; recall \cref{defEren}. Further,
since $\chi_{B_{\UV}}(\omega+\psi)\beta=v_{\UV}$, the expression in the second line is equal to 
$- u_{\UV,t}-tE_{\UV}^{\ren}$; see \cref{defuUV}.  
\end{proof}

Employing the formulas for the complex action in \cref{def:action,foruUVIto} it is
possible to derive the bounds and convergence relations of the next lemma, whose proof can be found in
\cite[\textsection6]{HinrichsMatte.2023}.
In fact, the second members on the right hand sides of \cref{def:action,foruUVIto} can be estimated trivially.
The main abstract ingredients used to deal with the stochastic integrals in \cref{def:action,foruUVIto}
are Kunita's inequality and an exponential tail estimate for L\'{e}vy type stochastic integrals
due to Applebaum and Siakalli \cite{Applebaum.2009,Siakalli.2019}.

\begin{lem}[{\bf Exponential moment bound and convergence}]\label{lem:expmb}
Let $p\in[1,\infty)$.
Then there exists $a_p\in(0,\infty)$, also depending on the model parameters $m_{\p}$,
$m_{\bos}$ and $g$, such that
\begin{align*}
\sup_{\UV\in[0,\infty]}\EE\Big[\sup_{s\in[0,t]}\eul^{pu_{\UV,s}}\Big]&\le \eul^{a_p(1+t)},\quad t\ge0.
\end{align*}
Furthermore, 
\begin{align*}
\EE\Big[\sup_{s\in[0,t]}\big|\eul^{u_{\UV,s}}-\eul^{u_{\infty,s}}\big|^p\Big]\xrightarrow{\;\;\UV\to\infty\;\;}0,\quad t\ge0.
\end{align*}
\end{lem}

The next result we shall apply without detailed proof is the flow relation \cref{floweq00}
implied by \cite[Lemma~7.9]{HinrichsMatte.2023}.
Since we only consider finite $\UV$ in \cref{floweq00}, its proof is, however, fairly elementary:
Applying both sides of \cref{floweq00} to an exponential vector in Fock space, i.e.,
a vector of the form $\expv{h}\coloneq F_1(h)(1,0,0,\ldots\,)$, the proof is reduced 
to three relations involving the integral processes $u_{\UV}$ and $U^\pm_{\UV}$ that can be verified by straightforward
substitutions. In fact, these computations are virtually identical to those
in the proof of \cite[Lemma~4.18]{MatteMoller.2018}.

For all $\UV\in[0,\infty]$, we denote by 
\begin{align}\label{relWst}
	W_{\UV,s,s+r}(x)&=\Gamma(e_{x})W_{\UV,s,s+r}(0)\Gamma(e_{-x}),\quad r,s\ge0,\,x\in\RR^2,
\end{align} 
the $\LO(\Fock)$-valued random variables obtained 
by working on the filtered probability space $(\Omega,\fr{F},(\fr{F}_{s+r})_{r\ge0},\PP)$ and
putting $(X_{s+r}-X_s)_{r\ge0}$ in place of $X$ in \cref{defUplusminus,defuUV,def:action,def:Wt}. 

\begin{lem}[{\bf Flow relation}]\label{lem:flow00}
Let $\UV\in[0,\infty)$. Then
\begin{align}\label{floweq00}
W_{\UV,t}(0)&=W_{\UV,s,t}(X_s)W_{\UV,s}(0),\quad t\ge s\ge 0.
\end{align}
\end{lem}

Finally, we shall need an integral equation involving $W_{\UV,t}(0)$ with finite $\UV$ applied
to a vector in the dense subset $\scr{C}$ of $\Fock$ given by
\begin{align}\label{defsscrC}
\scr{C}&\coloneq\mathrm{span}\big\{\expv{f}\in\Fock\,\big|\:f\in\dom(\omega^2)\big\}\subset\dom(\Id\Gamma(\omega)^2).
\end{align}
We shall abbreviate
\begin{align*}
h_{\UV}(x)&\coloneq \Id\Gamma(\omega)+\vp(e_{x}v_{\UV})+E_{\UV}^{\ren},\quad x\in\RR^2,
\end{align*}
so that 
\begin{align}\label{whHxih0}
\wh{H}_{\UV}(\xi)=\psi(\xi-\Id\Gamma(K))+h_{\UV}(0),\quad \xi\in\RR^2,\,\UV\in[0,\infty).
\end{align}

\begin{lem}[{\bf Integral equation}]\label{lem:inteqW0}
Let $\UV\in[0,\infty)$ and $\phi\in\scr{C}$. Then,
at every fixed elementary event,
$W_{\UV,s}(0)\phi\in\dom(\Id\Gamma(\omega))$ for all $s\ge0$, the path
$[0,\infty)\ni s\mapsto h_{\UV}(X_s)W_{\UV,s}(0)\phi\in\Fock$ is c\`{a}dl\`{a}g and
\begin{align}\label{IEQW}
W_{\UV,t}(0)\phi-\phi
&=
-\int_0^th_{\UV}(X_s)W_{\UV,s}(0)\phi\Id s,\quad t\ge0.
\end{align}
\end{lem}

\begin{proof}[Sketch of the proof of \cref{lem:inteqW0}.]
The first two statements hold in view of
\begin{align}\label{forWexpv}
W_{\UV,t}(0)\expv{f}&=\eul^{u_{\UV,t}-\langle U_{\UV,t}^-|f\rangle}\expv{\eul^{-t\omega}f-U^+_{\UV,t}},
\quad t\ge0,\,f\in\dom(\omega^2).
\end{align} 
Moreover, after scalar-multiplying \cref{forWexpv}
with an exponential vector $\expv{f_1}$ with $f_1\in\dom(\omega)$, the proof of \cref{IEQW} is
reduced to a straightforward computation making use of \cref{IEUplus};
see \cite[Lemma~4.2]{HinrichsMatte.2023} for details.
\end{proof}

\section{Feynman--Kac integrands and semigroups for fixed total momentum}\label{sec:FKint}

\noindent
In the whole \cref{sec:FKint} we fix $\UV\in[0,\infty]$. 
We shall discuss the Feynman--Kac integrands given by
\begin{align}\label{defwhWUV}
\wh{W}_{\UV,t}(\xi)&\coloneq \eul^{-\ii\xi\cdot X_t}\Gamma(e_{-X_t})W_{\UV,t}(0),\quad t>0,\,\xi\in\RR^2,
\end{align}
and $\wh{W}_{\UV,0}(\xi)\coloneq\id_{\Fock}$, as well as associated semigroups.

\begin{rem}\label{rem:whWmeascont}
For every $t>0$, we have the alternative formulas
\begin{align*}
\wh{W}_{\UV,t}(0)&=\eul^{u_{\UV,t}}F_{t/2}(-e_{-X_t}U^+_{\UV,t})\Gamma(e_{-X_t})F_{t/2}(-U^-_{\UV,t})^*
\\
&=\eul^{u_{\UV,t}}F_{t/3}(-e_{-X_t}U^+_{\UV,t})\eul^{\ii K\cdot X_t-t\Id\Gamma(\omega)/3}F_{t/3}(-U^-_{\UV,t})^*.
\end{align*}
We further know from \cite{GueneysuMatteMoller.2017} that the map $(0,\infty)\times L^2(\RR^2)\ni(s,h)\mapsto F_s(h)\in\LO(\Fock)$
is continuous, and since $|k|\le\omega(k)$, $k\in\RR^2$, the map
$(0,\infty)\times\RR^2\ni(s, x)\mapsto \eul^{\ii K\cdot x-s\Id\Gamma(\omega)/3}\in\LO(\Fock)$
is continuous as well. In conjunction with the separability of $L^2(\RR^2)$ as well as the
adaptedness and path regularity properties of $u_{\UV}$ and $U^\pm_{\UV}$ 
these remarks reveal the following for every $\xi\in\RR^2$:
\begin{enumerate}
\item[(a)] For every fixed $t\ge0$, $\wh{W}_{\UV,t}(\xi)$ is an $\fr{F}_t$-measurable and separably valued
$\LO(\Fock)$-valued random variable.
\item[(b)] At every fixed elementary event, the map $t\mapsto \wh{W}_{\UV,t}(\xi)\in\LO(\Fock)$ 
is right-continuous on $(0,\infty)$ and it has left limits at every point of $(0,\infty)$ that we
denote by $\wh{W}_{\UV,t-}(\xi)$, $t>0$.
\end{enumerate}
\end{rem}

\cref{lem:expmb} is the main ingredient for the next statement.
\begin{prop}[{\bf Moment bounds and convergence}]\label{prop:mb}
Let $p\in[1,\infty)$.
Then there exists some $c_p\in(0,\infty)$, solely depending on $p$ and the model parameters $m_{\p}$,
$m_{\bos}$ and $g$, such that
\begin{align}\label{mb}
\sup_{\xi\in\RR^2}\EE\Big[\sup_{s\in[0,t]}\|\wh{W}_{\UV,s}(\xi)\|^p\Big]
=\EE\Big[\sup_{s\in[0,t]}\|W_{\UV,s}(0)\|^p\Big]&\le \eul^{c_p(1+t)},\quad t\ge0.
\end{align}
Furthermore, 
\begin{align}\label{convW1}
\sup_{\xi\in\RR^2}
\EE\Big[\sup_{s\in[0,t]}\|\wh{W}_{\UV,s}(\xi)-\wh{W}_{\infty,s}(\xi)\|^p\Big]\xrightarrow{\;\;\UV\to\infty\;\;}0,\quad t\ge0.
\end{align}
\end{prop}
\begin{rem}
	The norms under the expectations in \cref{mb,convW1} actually are $\xi$-independent.
\end{rem}
\begin{proof}
Since $\Gamma(e_x)$ is unitary for every $x\in\RR^2$, the first relation in \cref{mb} is obvious from \cref{defwhWUV}.
The second relation in \cref{mb} and \cref{convW1} are implied by \cref{bdFt,lem:expmb} and the elementary bound
$\|\chi_{B_{\UV}\setminus B_\sigma}U_{\infty,t}^\pm\|_{t/2}^2\le 6\pi g^2(\sigma^2+m_{\bos}^2)^{-1/2}$ valid for all
$0\le\sigma<\UV\le\infty$.
\end{proof}

We also have an analogue of \cref{lem:flow00} for fixed total momentum.
\begin{prop}[{\bf Flow equation}]\label{propflow}
Let $\xi\in\RR^2$ and put
\begin{align*}
\wh{W}_{\UV,s,t}(\xi)&\coloneq \eul^{-\ii\xi\cdot (X_{t}-X_s)}\Gamma(e_{-(X_t-X_s)})W_{\UV,s,t}(0),
\quad t\ge s.
\end{align*}
Then, $\PP$-a.s.,
\begin{align}\label{floweq}
\wh{W}_{\UV,t}(\xi)&=\wh{W}_{\UV,s,t}(\xi)\wh{W}_{\UV,s}(\xi),\quad t\ge s.
\end{align}
\end{prop}

\begin{proof}
For finite $\UV$, \cref{floweq} is equivalent to \cref{floweq00} in view of \cref{defwhWUV,relWst}.
By virtue of \cref{convW1} and its analogue for $(\wh{W}_{\UV,s,s+r}(\xi))_{r\ge0}$, 
\cref{floweq} extends to $\UV=\infty$.
\end{proof}

In view of item (a) in \cref{rem:whWmeascont} as well as \cref{prop:mb}
the following $\LO(\Fock)$-valued expectations are well-defined:
\begin{align}\label{defwhT}
\wh{T}_{\UV,t}(\xi)&\coloneq\EE[\wh{W}_{\UV,t}(\xi)^*],\quad t\ge0,\,\xi\in\RR^2.
\end{align}

\begin{prop}[{\bf Norm bound and convergence}]
With $c_1$ denoting the $\UV$-independent constant appearing in \cref{prop:mb}, we  have
\begin{align}\label{opnbwhT}
\sup_{\xi\in\RR^2}\|\wh{T}_{\UV,t}(\xi)\|&\le \eul^{c_1(1+t)},\quad t\ge0.
\end{align}
Furthermore,
\begin{align}\label{univconvwhT}
\sup_{\xi\in\RR^2}\sup_{s\in[0,t]}\|\wh{T}_{\UV,s}(\xi)-\wh{T}_{\infty,s}(\xi)\|\xrightarrow{\;\;\UV\to\infty\;\;}0,
\quad t\ge0.
\end{align}
\end{prop}

\begin{proof}
Manifestly, \cref{opnbwhT,univconvwhT} follow from \cref{mb,convW1}, respectively.
\end{proof}

We continue by deriving a Markov property involving the family $(\wh{T}_{\UV,t}(\xi))_{t\ge0}$ 
that directly will entail its semigroup property.

\begin{thm}[{\bf Markov property}]\label{propMarkov}
Let $\xi\in\RR^2$ and $t\ge s\ge0$. Then, $\PP$-a.s.,
\begin{align}\label{markov1}
\EE^{\fr{F}_s}[\wh{W}_{\UV,t}(\xi)^*]&=\wh{W}_{\UV,s}(\xi)^*\wh{T}_{\UV,t-s}(\xi).
\end{align}
\end{thm}

\begin{proof}
This follows upon taking adjoints on both sides of \cref{floweq} and observing that
$\wh{W}_{\UV,s,t}(\xi)^*$ is $\fr{F}_s$-independent
while $\wh{W}_{\UV,s}(\xi)^*$ is $\fr{F}_s$-measurable. In fact,
let $(\fr{F}^s_r)_{r\ge0}$ denote the (automatically right-continuous) completion of the
natural filtration associated with $(X_{r+s}-X_s)_{r\ge0}$.
Applying \cref{rem:whWmeascont} to that filtration and the time-shifted L\'{e}vy process,
we see that $(\wh{W}_{\UV,s,s+r}(\xi)^*)_{r\ge0}$ is adapted to $(\fr{F}^s_r)_{r\ge0}$.
Since $X$ is $(\fr{F}_r)_{r\ge0}$-L\'{e}vy, we know, however, that
each $\fr{F}^s_r$ with $r\ge0$ and $\fr{F}_s$ are independent.
To get \cref{markov1} we also exploit that $\wh{W}_{\UV,s,t}(\xi)^*$ and
$\wh{W}_{\UV,t-s}(\xi)^*$ have the same distribution.
\end{proof}

Taking expectations in \cref{markov1} with $t=r+s$ we arrive at the following result:

\begin{cor}[{\bf Semigroup property}]
For all $\xi\in\RR^2$ and $r,s\ge0$,
\begin{align*}
\wh{T}_{\UV,s+r}(\xi)&=\wh{T}_{\UV,s}(\xi)\wh{T}_{\UV,r}(\xi).
\end{align*}
\end{cor}

\section{Feynman--Kac formulas for fiber Hamiltonians}\label{sec:FKxi}

\noindent
We now turn to the derivation of the Feynman--Kac formulas for the fiber Hamiltonians.
First, we shall do this for the ultraviolet regularized operators,
by showing that the semigroup $(\wh{T}_{\UV,t}(\xi))_{t\ge0}$ is strongly continuous
and identifying $\wh{H}_{\UV}(\xi)$ as its generator. For the latter two tasks we require
the stochastic differential equations derived in the next lemma. Recall
the definition \cref{defsscrC} of $\scr{C}$.

\begin{lem}[{\bf Stochastic differential equation with cutoff}]\label{thm:SDEUV}
Let $\UV\in[0,\infty)$, $\xi\in\RR^2$ and $\phi\in\scr{C}$. Then, $\PP$-a.s.,
\begin{align*}
\wh{W}_{\UV,t}(\xi)\phi-\phi&=-\int_0^t\wh{H}_{\UV}(\xi)\wh{W}_{\UV,s}(\xi)\phi\Id s
\\
&\quad+\int_{(0,t]\times\RR^2}(\eul^{-\ii\xi\cdot z}\Gamma(e_{-z})-1)\wh{W}_{\UV,s-}(\xi)\phi\Id\wt{N}(s,z),
\quad t\ge0.
\end{align*}
\end{lem}

\begin{proof}
Let $\eta\in\dom(\Id\Gamma(\omega)^2)$.
Since $\RR^2\ni z\mapsto \eul^{\ii (\xi-\Id\Gamma(K))\cdot z}\eta\in\Fock$ is twice continuously
differentiable and bounded with bounded first and second order partial derivatives, we $\PP$-a.s. have
the It\^{o} formula
\begin{align*}
&\eul^{\ii (\xi-\Id\Gamma(K))\cdot X_{t}}\eta-\eta
\\
&=-\int_0^t\psi(\xi-\Id\Gamma(K))\eul^{\ii (\xi-\Id\Gamma(K))\cdot X_{s}}\eta\Id s
\\
&\quad+\int_{(0,t]\times\RR^2}
\eul^{\ii (\xi-\Id\Gamma(K))\cdot X_{s-}}(\eul^{\ii(\xi-\Id\Gamma(K))\cdot z }-1)\eta\Id\wt{N}(s,z),\quad t\ge0.
\end{align*}
In conjunction with \cref{lem:inteqW0} and It\^{o}'s product rule for scalar products it $\PP$-a.s. implies
\begin{align*}
&\langle \eul^{\ii (\xi-\Id\Gamma(K))\cdot X_{t}}\eta|W_{\UV,t}(0)\phi\rangle-\langle\eta|\phi\rangle
\\
&=-\int_0^t\langle \eul^{\ii (\xi-\Id\Gamma(K))\cdot X_{s}}\eta|(\psi(\xi-\Id\Gamma(K))
+h_{\UV}(X_s))W_{\UV,s}(0)\phi\rangle\Id s
\\
&+\int_{(0,t]\times\RR^2}\langle
\eul^{\ii (\xi-\Id\Gamma(K))\cdot X_{s-}}(\eul^{\ii(\xi-\Id\Gamma(K))\cdot z }-1)\eta|
W_{\UV,s-}(0)\phi\rangle\Id\wt{N}(s,z),
\end{align*}
for all $t\ge0$. On account of \cref{comrelGammavp,defwhWUV} we further have
\begin{align*}
\eul^{-\ii (\xi-\Id\Gamma(K))\cdot X_{s}}h_{\UV}(X_s)W_{\UV,s}(0)\phi
&=h_{\UV}(0)\wh{W}_{\UV,s}(\xi)\phi,\quad s\ge0,
\end{align*}
as well as $\eul^{-\ii (\xi-\Id\Gamma(K))\cdot X_{s-}}W_{\UV,s-}(0)=\wh{W}_{\UV,s-}(\xi)$, $s>0$.
Since $\eta$ can be chosen from a countable dense subset,
these remarks and \cref{whHxih0} $\PP$-a.s. imply the asserted stochastic differential equation.
\end{proof}

We will also need the following bound.
In its proof we argue similarly as in the proof of \cite[Lemma~10.9]{GueneysuMatteMoller.2017}.

\begin{lem}\label{lem:sc0}
Let $\UV\in[0,\infty)$ and $\xi\in\RR^2$.
Then there exists a constant $b_{\UV}(\xi)\in(0,\infty)$, also depending on the model
parameters $m_{\p}$, $m_{\bos}$ and $g$, such that
\begin{align}\label{sc0}
\EE\big[\|(1+\Id\Gamma(\omega))^{-1}(\wh{W}_{\UV,t}(\xi)\phi-\phi)\|^2\big]
&\le b_{\UV}(\xi)t\eul^{b_{\UV}(\xi)t}\|\phi\|^2,\quad t\ge0,\,\phi\in\Fock.
\end{align}
\end{lem}

\begin{proof}
To start with we assume that $\phi\in\scr{C}$.
Abbreviating $\theta\coloneq 1+\Id\Gamma(\omega)$ and $\eta_t\coloneq\theta^{-1}(\wh{W}_{\UV,t}(\xi)\phi-\phi)$, $t\ge0$,
we then infer from \eqref{thm:SDEUV} and It\^{o}'s product formula that, $\PP$-a.s.,
\begin{align}\nonumber
\|\eta_t\|^2&=-2\int_0^t\Re\langle\eta_s | \theta^{-1}\wh{H}_{\UV}(\xi)\wh{W}_{\UV,s}(\xi)\phi\rangle\Id s
\\\nonumber
&\quad +\int_{(0,t]\times\RR^2}\|\theta^{-1}(\eul^{-\ii\xi\cdot z}\Gamma(e_{-z})-1)\wh{W}_{\UV,s-}(\xi)\phi\|^2\Id N(s,z)
\\\label{sc0b}
&\quad +2\int_{(0,t]\times\RR^2}\Re\langle\eta_{s-}|
\theta^{-1}(\eul^{-\ii\xi\cdot z}\Gamma(e_{-z})-1)\wh{W}_{\UV,s-}(\xi)
\phi\rangle\Id\wt{N}(s,z),\quad t\ge0,
\end{align}
where $N$ is the Poisson point measure defined by the jumps of $X$. Here 
\begin{align*}
\|\theta^{-1}(\eul^{-\ii\xi\cdot z}\Gamma(e_{-z})-1)\|\le \min\{|z|,2\}(1+|\xi|),\quad z\in\RR^2,
\end{align*} 
and in view of \cref{rbvp}
we know that the operator $\theta^{-1}\wh{H}_{\UV}(\xi)$ is bounded.
We further recall \cref{mb}, from which we infer the {\em a priori} bound
$\EE[\sup_{s\in[0,t]}\|\eta_s\|^4]\le \eul^{4c(1+t)}\|\phi\|^4$, $t\ge0$, with $c\in(0,\infty)$ solely depending
on $m_{\p}$, $m_{\bos}$ and $g$. On account of these bounds, the stochastic integral in the last line
of \cref{sc0b} is a martingale starting at $0$ and in particular its expectation is $0$. Furthermore, the expectation
of the $N$-integral in the second line of \cref{sc0b} equals
\begin{align*}
\int_0^t\int_{\RR^2}\EE\big[\|\theta^{-1}(\eul^{-\ii\xi\cdot z}\Gamma(e_{-z})-1)\wh{W}_{\UV,s-}(\xi)\phi\|^2\big]\Id\nu(z)\,\Id s.
\end{align*}
Upon taking expectations on both sides of \cref{sc0b}  we thus find
\begin{align*}
\EE[\|\eta_t\|^2]&\le 2t\eul^{(c+c_2/2)(1+t)}\|\theta^{-1}\wh{H}_{\UV}(\xi)\|\|\phi\|^2
\\
&\quad +t\eul^{c_2(1+t)}(1+|\xi|)^2\int_{\RR^2}\min\{|z|,2\}^2\Id\nu(z)\,\|\phi\|^2,\quad t\ge0.
\end{align*}
Here the $\nu$-integral is finite because $\nu$ is a L\'{e}vy measure.
Finally, we invoke the dominated convergence theorem and \cref{mb} to extend \cref{sc0}
to general $\phi\in\Fock$.
\end{proof}

We can now prove the Feynman--Kac formula for the ultraviolet regularized fiber Hamiltonians.

\begin{thm}[{\bf Feynman--Kac formula with cutoff}]\label{thm:FKxiUV}
Let $\UV\in[0,\infty)$ and $\xi\in\RR^2$. Then 
\begin{align}\label{sc1}
\lim_{t\to s}\|(\wh{T}_{\UV,t}(\xi)-\wh{T}_{\UV,s}(\xi))(1+\Id\Gamma(\omega))^{-1}\|&=0,\quad s\ge0,
\end{align}
and in particular the semigroup $(\wh{T}_{\UV,t}(\xi))_{t\ge0}$ is
strongly continuous. Furthermore,
\begin{align}\label{FKxiUV}
\eul^{-t\wh{H}_{\UV}(\xi)}&=\wh{T}_{\UV,t}(\xi),
\end{align}
and in particular $\wh{T}_{\UV,t}(\xi)$ is selfadjoint for every $t\ge0$.
\end{thm}

\begin{proof}
First, we prove \cref{sc1} which together with \cref{opnbwhT} 
entails strong continuity. To that end it suffices to show that
\begin{align}\label{sc2}
\lim_{t\downarrow0}\|(\wh{T}_{\UV,t}(\xi)-\id_{\Fock})(1+\Id\Gamma(\omega))^{-1}\|&=0,
\end{align}
by the semigroup property and \cref{opnbwhT}.
Using Cauchy--Schwarz inequalities and applying \cref{lem:sc0} we find, however,
\begin{align*}
&\|(\wh{T}_{\UV,t}(\xi)-\id_{\Fock})(1+\Id\Gamma(\omega))^{-1}\|
\\
&= \sup_{\|\phi_1\|=\|\phi_2\|=1}
\big|\EE[\langle(1+\Id\Gamma(\omega))^{-1}(\wh{W}_{\UV,t}(\xi)\phi_1-\phi_1)|\phi_2\rangle]\big|
\\
&\le \sup_{\|\phi_1\|=1}
\EE\big[\|(1+\Id\Gamma(\omega))^{-1}(\wh{W}_{\UV,t}(\xi)\phi_1-\phi_1)\|^2\big]^{1/2}
\le (b_{\UV}(\xi)t\eul^{b_{\UV}(\xi)t})^{1/2},\quad t>0,
\end{align*}
which proves \cref{sc2}, of course.

By our results proven so far, we know that the semigroup $(\wh{T}_{\UV,t}(\xi))_{t\ge0}$ has a closed
generator, call it $G_{\UV}(\xi)$, whose spectrum is contained in the half-space
$\{z\in\CC|\,\Re[z]\ge a\}$ for some $a\in\RR$. We shall now show that $G_{\UV}(\xi)=\wh{H}_{\UV}(\xi)$,
which is equivalent to the validity of \cref{FKxiUV} for all $t\ge0$. In fact, it suffices to show the inclusion
$\wh{H}_{\UV}(\xi)\subset G_{\UV}(\xi)$, because we then can pick some $\zeta\in\CC$ belonging
to the resolvent sets of both $\wh{H}_{\UV}(\xi)$ and $G_{\UV}(\xi)$ (e.g., $\zeta=a-1+\ii$)
and apply the second resolvent identity to see that $(\wh{H}_{\UV}(\xi)-\zeta)^{-1}=(G_{\UV}(\xi)-\zeta)^{-1}$.

So let $\eta\in\dom(\wh{H}_{\UV}(\xi))=\dom(\Id\Gamma(\omega))$.
Scalar-multiplying the SDE in \cref{thm:SDEUV} with $\eta$ and
taking expectations afterwards, we find
\begin{align*}
\langle\wh{T}_{\UV,t}(\xi)\eta-\eta|\phi\rangle&=-\int_0^t\langle\wh{T}_{\UV,s}(\xi)\wh{H}_{\UV}(\xi)\eta|\phi\rangle\Id s
+\EE[M_t(\eta,\phi)],
\end{align*}
for all $t\ge0$ and $\phi\in\scr{C}$. Here the stochastic integral process given by
\begin{align*}
M_t(\eta,\phi)&\coloneq 
\int_{(0,t]\times\RR^2}\langle \eta|(\eul^{-\ii\xi\cdot z}\Gamma(e_{-z})-1)\wh{W}_{\UV,s-}(\xi)\phi\rangle\Id\wt{N}(s,z),
\quad t\ge0,
\end{align*}
is a martingale starting at $0$. This follows from \cref{mb} and the bound
\begin{align*}
\|(\eul^{-\ii\xi\cdot z}\Gamma(e_{-z})-1)^*\eta\|&\le \min\{|z|,2\}(2\|\eta\|+|\xi|\|\eta\|+\|\Id\Gamma(\omega)\eta\|),
\quad z\in\RR^2.
\end{align*}
In particular $\EE[M_t(\eta,\phi)]=0$, $t\ge0$. Since $\phi$ can be chosen
in a dense subset of $\Fock$, we deduce that
\begin{align}\label{gen1}
\frac{1}{t}(\wh{T}_{\UV,t}(\xi)\eta-\eta)=-\frac{1}{t}\int_0^t\wh{T}_{\UV,s}(\xi)\wh{H}_{\UV}(\xi)\eta\Id s,
\quad t>0,
\end{align}
with an $\Fock$-valued Bochner-Lebesgue integral on the right hand side.
The whole expression on the right hand side of \cref{gen1} converges to $-\wh{H}_{\UV}(\xi)\eta$, as $t\downarrow0$,
because $\wh{T}_{\UV,s}(\xi)\wh{H}_{\UV}(\xi)\eta\to\wh{H}_{\UV}(\xi)\eta$, as $s\downarrow0$,
by strong continuity of the semigroup. Thus, $\eta\in\dom(G_{\UV}(\xi))$
with $G_{\UV}(\xi)\eta=\wh{H}_{\UV}(\xi)\eta$.
\end{proof}

Using the convergence statements proven in \cref{sec:FKint} it is not hard to deduce our main result for the fiber Hamiltonians.

\begin{thm}[{\bf Renormalization; Feynman--Kac formula without cutoff}]\label{thm:renFKxi}
${}$
Let $\xi\in\RR^2$. Then the following holds:
\begin{enumerate}
\item[{\rm(i)}] Statement \cref{sc1} holds for $\UV=\infty$ as well. 
\item[{\rm(ii)}] $(\wh{T}_{\infty,t}(\xi))_{t\ge0}$ is a strongly continuous semigroup of selfadjoint
operators satisfying $\|\wh{T}_{\infty,t}(\xi)\|\le\eul^{c(1+t)}$ for all $t\ge0$ and some $c\in(0,\infty)$.
\item[{\rm(iii)}] Denote by $\wh{H}(\xi)$ the selfadjoint, lower semibounded generator
of $(\wh{T}_{\infty,t}(\xi))_{t\ge0}$, so that 
\begin{align}\label{FKxi}
\eul^{-t\wh{H}(\xi)}=\wh{T}_{\infty,t}(\xi)=\EE[\wh{W}_{\infty,t}(\xi)^*],\quad t\ge0.
\end{align}
Then $\wh{H}_{\UV}(\xi)$ converges in the norm resolvent sense to $\wh{H}(\xi)$ as $\UV\to\infty$.
\end{enumerate}
\end{thm}

\begin{proof}
Part (i) is a consequence of \cref{sc1} and the uniform convergence on compact time intervals
in \cref{univconvwhT}. Strong continuity of $(\wh{T}_{\infty,t}(\xi))_{t\ge0}$ follows from (i) and \cref{opnbwhT}.
Each $\wh{T}_{\infty,t}(\xi)$ with $t\ge0$ is selfadjoint since by \cref{univconvwhT} it is the norm limit as $\UV\to\infty$
of the selfadjoint operators $\wh{T}_{\UV,t}(\xi)$; recall the last statement of \cref{thm:FKxiUV}.
The norm bounds in (ii) have already been stated in \cref{opnbwhT}.
By (ii) and the Hille--Yosida theorem, an operator $\wh{H}(\xi)$ as in (iii) exists and is unique.
The norm resolvent convergence $\wh{H}_{\UV}(\xi)\to \wh{H}(\xi)$, $\UV\to\infty$, is known to be
equivalent to the norm convergence $\eul^{-t\wh{H}_{\UV}(\xi)}\to \eul^{-t\wh{H}(\xi)}$, $\UV\to\infty$,
for every $t\ge0$. The latter holds due to \cref{univconvwhT,FKxiUV,FKxi}, which proves (iii).
\end{proof}

\begin{rem}
For all $\UV\in[0,\infty]$ and $t\ge0$, the map $\RR^2\ni\xi\mapsto \wh{T}_{\UV,t}(\xi)\in\LO(\Fock)$
is continuous. This follows from \cref{defwhWUV}, \cref{defwhT}, the $\PP$-integrability of
$\|W_{\UV,t}(0)\|$ and the dominated convergence theorem for the Bochner-Lebesgue integral.
In view of \cref{FKxi} we may conclude that the familiy $(\wh{H}(\xi))_{\xi\in\RR^2}$ is strongly
resolvent measurable and in particular its direct integral is a well-defined selfadjoint operator
in $L^2(\RR^2,\Fock)$.
\end{rem}

More is true for strictly positive particle masses:

\begin{rem}
Let $\UV\in[0,\infty]$ and $t\ge0$. Assume that $m_{\p}>0$ and set $S(m_{\p})\coloneq\{z\in\CC^2|\,|\Im[z]|<m_{\p}\}$.
Then the expectations
\begin{align*}
\wh{T}_{\UV,t}(\zeta)&\coloneq \EE[\eul^{\ii \zeta\cdot X_t}\wh{W}_{\UV,t}(0)^*],\quad \zeta\in S(m_{\p}),
\end{align*}
are well-defined and extend the previously considered family $(\wh{T}_{\UV,t}(\xi))_{\xi\in\RR^2}$ to  $S(m_{\p})$. Moreover, the map $S(m_{\p})\ni\zeta\mapsto\wh{T}_{\UV,t}(\zeta)\in\LO(\Fock)$ is analytic.
This follows easily from H\"{o}lder's inequality and \cref{mb} since $\EE[\eul^{|\zeta||X_t|}]<\infty$
whenever $|\zeta|<m_{\p}$.
\end{rem}

\section{The full Hamiltonian revisited}\label{sec:renrev}

\noindent
We wish to verify that the renormalized operators $\wh{H}(\xi)$, $\xi\in\RR^2$, give rise
to a fiber decomposition of the renormalized full Hamiltonian.
In what follows $\scr{U}$ again denotes the Lee-Low-Pines tranformation of \cref{sec:LLP}.
The next corollary actually provides an independent existence proof for the norm resolvent
limit of the family $(H_{\UV})_{\UV\in[0,\infty)}$, based on the key ingredients 
collected in \cref{sec:keys}:

\begin{cor}\label{cor:nrcHUV}
As $\UV$ tends to infinity, $\scr{U}H_{\UV}\scr{U}^*$ converges in the norm resolvent sense to
$\int_{\RR^2}^\oplus \wh{H}(\xi)\Id\xi$.
\end{cor}

\begin{proof}
On account of \cref{fibdecHUV} and the Feynman--Kac formulas \cref{FKxiUV,FKxi}, the statement is 
equivalent to the operator norm convergences
\begin{align*}
\int_{\RR^2}^\oplus \wh{T}_{\UV,t}(\xi)\Id\xi\xrightarrow{\;\;\UV\to\infty\;\;}\int_{\RR^2}^\oplus \wh{T}_{\infty,t}(\xi)\Id\xi,
\quad t>0,
\end{align*}
which follow from the $\xi$-uniform convergence in \cref{univconvwhT}.
\end{proof}

With $H$ denoting the norm resolvent limit of $(H_{\UV})_{\UV\in[0,\infty)}$ we thus arrive at:

\begin{cor}\label{thm:fibH}
$\scr{U}H\scr{U}^*=\int_{\RR^2}^\oplus\wh{H}(\xi)\Id\xi$.
\end{cor}

Finally, we fulfill a promise we gave at the end of \cref{sec:FKH}:

\begin{proof}[Alternative proof of the Feynman--Kac formula \cref{eq:FK}.]
Let $t>0$.
We assume that $\Psi\in L^2(\RR^2,\Fock)$ has an integrable Fourier transform $\wh{\Psi}$
and set $\Phi(x)\coloneq\Gamma(e_x)\Psi(x)$, a.e. $x\in\RR^2$. In view of
\begin{align*}
\eul^{-tH_{\UV}}\Phi
&=\scr{U}^*\int_{\RR^2}^\oplus\wh{T}_{\UV,t}(\xi)\Id\xi\,\scr{U}\Phi
\end{align*}
as well as \cref{defLLP,defwhWUV,defwhT} we find
\begin{align*}
(\eul^{-tH_{\UV}}\Phi)(x)&=\Gamma(e_x)\frac{1}{2\pi}\int_{\RR^2}
\eul^{\ii \xi\cdot x}\EE\big[\eul^{\ii\xi\cdot X_t}W_{\UV,t}(0)^*\Gamma(e_{X_t})\big]\wh{\Psi}(\xi)\Id\xi
\\
&=\EE\big[\Gamma(e_x)W_{\UV,t}(0)^*\Gamma(e_{-x})\Gamma(e_{x+X_t})\Psi(x+X_t)\big]
\\
&=\EE[W_{\UV,t}(x)^*\Phi(x+X_t)],\quad \text{a.e. $x$,}
\end{align*}
where we applied the Fubini and Fourier inversion theorems in the second step
and \cref{def:Wt} in the third one. This proves \cref{eq:FK} for all $\Phi$ in a dense subset of
$L^2(\RR^2,\Fock)$. Since, by \cref{def:Wt,mb},
\begin{align*}
\int_{\RR^2}\|\EE[W_{\UV,t}(x)^*\Phi(x+X_t)]\|^2\Id x
&\le \eul^{c_2(1+t)}\int_{\RR^2}\EE[\|\Phi(x+X_t)\|^2]\Id x
\\
&=\eul^{c_2(1+t)}\|\Phi\|^2,\quad\Phi\in L^2(\RR^2,\Fock),
\end{align*}
it is clear that \cref{eq:FK} extends to all $\Phi\in L^2(\RR^2,\Fock)$ by approximation.
\end{proof}


\subsection*{Acknowledgements}
The authors thank the Research Institute for Mathematical Sciences in Kyoto, Kyushu University and especially Fumio Hiroshima for their support and generous hospitality during and after the RIMS Workshop {\em Mathematical aspects of quantum fields and related  topics} in January 2023.
	BH acknowledges support by the Ministry of Culture and Science of the State of North Rhine-Westphalia within the project PhoQC.


\bibliographystyle{halpha-abbrv}

\bibliography{\jobname}

\end{document}